\documentclass{llncs}

\usepackage{amsmath,amssymb,framed,enumerate}
\usepackage{algorithm}
\usepackage{algorithmic,eqparbox,array}

\DeclareMathOperator{\poly}{poly}

\newcommand{\comment}[1]{}
\newcommand{\ltsqn}[1]{\|#1\|^2}
\newcommand{\ltn}[1]{\|#1\|}
\newcommand{\dotp}[2]{\langle #1, #2 \rangle}
\newcommand{\svp}{{\tt SVP}}

\newcommand{\cvp}{{\tt CVP}}
\newcommand{\bdd}{{\tt BDD}}
\newcommand{\LLL}{{\tt LLL}}
\newcommand{\usvp}{{\tt uSVP}}
\newcommand{\Sp}{{\texttt{span}}}
\newcommand{\G}{{\texttt{Gap}}}

\newcommand{\B}{{\bf B}}
\newcommand{\bi}{{\bf b}}
\newcommand{\ti}{{\bf t}}
\newcommand{\zi}{{\bf z}}
\newcommand{\vi}{{\bf v}}
\newcommand{\ui}{{\bf u}}
\newcommand{\xl}{{\bf x}}

\newcommand{\R}{{\bf\mathbb{R}}}
\newcommand{\Z}{{\bf\mathbb{Z}}}
\newcommand{\di}{{\bf\texttt{d}}}
\newcommand{\La}{{\bf\mathbb{L}}}
\newcommand{\Or}{{\bf\mathfrak{O}}}
\newcommand{\Q}{{\bf\mathbb{Q}}}

\newcommand{\fb}{{\tt findbasis}}

\title{Approximating the Closest Vector Problem Using an Approximate Shortest Vector Oracle}

\author{Chandan Dubey\thanks{Chandan Dubey is partially supported by the Swiss National Science Foundation (SNF), project no. 200021-132508} \and Thomas Holenstein}
\institute{Institute for Theoretical Computer Science\\ETH Zurich\\\email{chandan.dubey@inf.ethz.ch}\\\email{thomas.holenstein@inf.ethz.ch}}

\begin{document}

\maketitle

\begin{abstract}
We give a polynomial time Turing reduction from the $\gamma^2\sqrt{n}$-approximate closest vector problem on a lattice of dimension $n$ to a $\gamma$-approximate oracle for the shortest vector problem. This is an improvement over a reduction by Kannan, which achieved $\gamma^2n^{\frac{3}{2}}$.  
\end{abstract}
%\begin{keywords}
%Lattice, Closest Vector Problem, Shortest Vector Problem, Polynomial time reduction, Bounded Distance Decoding.
%\end{keywords}

\section{Introduction}

A {\em lattice} is the set of all integer combinations of~$n$ linearly independent vectors $\bi_1, \bi_2, \dots, \bi_n$ in $\R^m$. These vectors are also referred to as a {\em basis} of the lattice. The {\em successive minima} $\lambda_i(\La)$ (where $i = 1, \dots, n$) for the lattice $\La$ are among the most fundamental parameters associated to a lattice. The value $\lambda_i(\La)$ is defined as the smallest $r$ such that a sphere of radius $r$ centered around the origin contains at least $i$ linearly independent lattice vectors. Lattices have been investigated by computer scientists for a few decades after the discovery of the LLL algorithm \cite{LLL82}. More recently, Ajtai \cite{Ajt96} showed that lattice problems have a very desirable property for cryptography: they exhibit a worst-case to average-case reduction.

We now describe some of the most fundamental and widely studied lattice problems. Given a lattice $\La$, the $\gamma$-approximate shortest vector problem ($\gamma$-$\svp$ for short) is the problem of finding a non-zero lattice vector of length at most $\gamma\lambda_1(\La)$. Let the minimum distance of a point $\ti \in \R^m$ from the lattice $\La$ be denoted by $\di(\ti,\La)$. Given a lattice $\La$ and a point $\ti \in \R^m$, the $\gamma$-approximate closest vector problem or $\gamma$-$\cvp$ for short is the problem of finding a $\vi \in \La$ such that $\ltn{\vi-\ti} \leq \gamma\di(\ti,\La)$. 

Besides the search version just described, $\cvp$ and $\svp$ also have a gap version. The problem $\G\cvp_\gamma(\B,\ti)$  asks the distance of $\ti$ from the lattice $\La(\B)$ within a factor of $\gamma$, and $\G\svp_\gamma(\B)$ asks for $\lambda_1(\B)$ within a factor of $\gamma$. 
This paper deals with the search version described above.

The problems $\cvp$ and $\svp$ are quite well studied. The ${\tt Gap}$ versions of the problems are arguably easier than their ${\tt search}$ counterparts. We know that $\cvp$ and $\svp$ can be solved exactly in deterministic $2^{O(n)}$ time \cite{MV10,AKS01}. In polynomial time they can be approximated within a factor of $2^{n(\log\log n)/\log n}$ using LLL \cite{LLL82} and subsequent improvements by Schnorr \cite{Sch87} and Micciancio et. al. \cite{MV10} (for details, see the book by Micciancio and Goldwasser \cite{GM02}). On the other hand, it is known that there exists $c>0$, such that no polynomial time algorithm can approximate $\G\cvp$ and $\G\svp$ within a factor of $n^{c/\log\log n}$, unless {\bf P} $=$ {\bf NP} or another unlikely scenario is true \cite{DKRS03,HR07}. The security of hardness of cryptosystems following Ajtai's seminal work \cite{Ajt96} is based on the worst-case hardness of $\tilde{O}(n^2)$-${\tt Gap}\svp$ \cite{Reg04,Pei09,LM09}. In the hardness area, $\cvp$ is much more understood than $\svp$. For example, as opposed to $\cvp$, until now all known {\bf NP}-hardness proofs for $\svp$ \cite{Ajt98,Mic01,Kho05,HR07} are randomized. A way to prove deterministic hardness of $\svp$ is to prove better reductions from $\cvp$ to $\svp$. This paper aims to study and improve the known relations between these two problems.

A very related result is from Kannan \cite{Kannan87}, who gave a way to solve $\sqrt{n}$-$\cvp$ using an exact $\svp$ oracle. A generalization of his reduction was used to solve $\cvp$ within a factor of $(1+\epsilon)$ by reducing it to sampling short vectors in the lattice \cite{AKS02}. The improvement from $\sqrt{n}$ to $(1+\epsilon)$ is achieved mainly because the reduction uses $2^{O(n)}$ time instead of polynomial. It is also known that a $\gamma$-$\cvp$ oracle can be used to solve $\gamma$-$\svp$ \cite{GMSS99}. 

In a survey \cite{Kan87}, Kannan gave a different reduction from $\gamma^2n^{\frac{3}{2}}$-$\cvp$ to $\gamma$-$\svp$. A few words of comparison between our methods and the method used by Kannan \cite{Kan87}. Kannan uses the dual lattice (denoted by $\B^{*} = (\B^T)^{-1}$, where $\B^T$ is the transpose of the matrix $\B$) and the transference bound $\lambda_1(\B)\lambda_1(\B^{*}) \leq n$ to find a candidate close vector. Due to the fact that he applies the $\svp$ oracle on both $\La$ as well as $\La^{*}$, he loses an additional factor of $n$. Our method does not use the dual lattice.

{\bf Our contribution:} We improve the result by Kannan \cite{Kan87}, which shows that $\gamma^2n^{3/2}$-$\cvp$ can be solved using an oracle to solve $\gamma$-$\svp$, and solve $\gamma^2\sqrt{n}$-$\cvp$ using the same oracle.

For this, we essentially combine the earlier result by Kannan \cite{Kannan87} 
with a reduction by Lyubashevsky and Micciancio \cite{LM09}, as we explain
now in some detail.

Our starting point is the earlier reduction by Kannan, which solves $\sqrt{n}$-$\cvp$ using an exact $\svp$-oracle.  In order to explain our ideas, we first shortly describe his reduction.  Given a $\cvp$-instance $\B \in \Q^{m\times n}, \ti \in \R^m$, Kannan uses the $\svp$-oracle to find $\lambda_1(\B)$.  He then creates the new basis $\tilde{\B}=\Bigg[\begin{array}{cc}\B & \ti\\0 & \alpha\end{array}\Bigg]$, where he picks $\alpha$ carefully somewhat smaller than $\lambda_1(\B)$.  Now, if $\di(\ti,\B)$ is significantly smaller than $\lambda_1(\B)$ (say, $\lambda_1(\B)/3$), then the shortest vector in $\tilde{\B}$ is $\Big[\begin{array}{c}\ti^\dagger-\ti\\-\alpha\end{array}\Big]$, where $\ti^\dagger$ is the lattice vector closest to $\ti$ (i.e., the vector we are trying to find).  On the other hand if $\di(\ti,\B)$ is larger than $\lambda_1(\B)/3$, then Kannan projects the instance in the direction orthogonal to the shortest vector of $\B$.  This reduces the dimension by $1$, and an approximation in the resulting instance can be used to get an approximation in the original instance, because the projected approximation can be ``lifted'' to find some original lattice point which is not too far from $\ti$.

We show that in case we only have an approximation oracle for $\svp$, 
we can argue as follows.
First, if $\di(\ti,\B) \leq \frac{\lambda_1(\B)}{2\gamma}$, then we have
an instance of a so called ``Bounded Distance Decoding'' problem.
By a result of Lyubashevsky and Micciancio
\cite{LM09}, this can be solved using the the
oracle we assume.
In case $\di(\ti,\B) > \frac{\lambda_1(\B)}{2\gamma}$ we
can recurse in the same way as Kannan does.  
The approximation factor $\gamma^2 \sqrt{n}$ comes from this case: 
lifting a projection after the recursion returns, incurs
an error of roughly the half the length of the vector $\vi$ which
was used to project.
Since this $\vi$ can have length almost $\gamma \lambda_1(\B)$, the
length of $\vi$ can be almost a factor $\gamma^2$ larger than
$\di(\ti,\B)$.
The squares of these errors then add up as in Kannan's reduction, 
which gives a total approximation
factor of $\gamma^2\sqrt{n}$.

We remark that even though we do not know which of the two cases
apply, we can simply run both, and then use the better result.

Finally, we would like to mention that to the best of our knowledge
there is no published proof that in Kannan's algorithm \cite{Kannan87}
the projected bases have a representation which is polynomial in the
input size. We show that this is indeed the case.  For this,
it is essentially enough to use
a lemma from~\cite{GM02} which states that the vectors in a 
Gram-Schmidt orthogonalization have this property.

\section{Preliminaries}

\subsection{Notation}
A lattice basis is a set of linearly independent vectors $\bi_1, \dots, \bi_n \in \R^m$.  It is sometimes convenient to think of the basis as an $n \times m$ matrix $\B$, whose $n$ columns are the vectors $\bi_1, \dots, \bi_n$. The lattice generated by the basis $\B$ will be written as $\La(\B)$ and is defined as $\La(\B) = \{\B x | x \in \Z^n\}$. The {\em span} of a basis $\B$, denoted as $\Sp(\B)$, is defined as $\{\B y | y \in \R^n\}$.  We will assume that the lattice is over rationals, i.e., $\bi_1, \dots, \bi_n \in \Q^m$, and the 
entries are represented by the pair of numerator and denominator.  An {\em elementary} vector $v \in \La(\B)$ is a vector which cannot be written as a non-trivial multiple of another lattice vector.  

A {\em shortest vector} of a lattice is a non-zero vector in the lattice whose $\ell_2$ norm is minimal.   The length of the shortest vector is $\lambda_1(\B)$, where $\lambda_1$ is as defined in the introduction.  For a vector $\ti \in \R^m$, let $\di(\ti,\La(\B))$ denote the distance of $\ti$ to the closest lattice point in $\B$.  We use $\ti^{\dagger}$ to denote a (fixed) closest vector to $\ti$ in $\La(\B)$. 

For two vectors $\ui$ and $\vi$ in $\R^m$, $\vi|_\ui$ denotes the component of $\vi$ in the direction of $\ui$ i.e., $\vi|_{\ui} = \frac{\dotp{\vi}{\ui}}{\dotp{\ui}{\ui}}\ui$. Also, the component of $\vi$ in the direction orthogonal to $\ui$ is denoted by $\vi_{\perp \ui}$ i.e., the vector $\vi-\vi|_{\ui}$.  

Consider a lattice $\La(\B)$ and a vector $\vi \in \La(\B)$ in the lattice. Then the projected lattice of $\La(\B)$ perpendicular to $\vi$ is $\La(\B_{\perp \vi}) := \{\ui_{\perp \vi}| \ui \in \La(\B)\}$.  A basis of $\La(\B_{\perp \bi_1})$ is given by the vectors $\{{\bi_2}_{\perp \bi_{1}}, \dots, {\bi_n}_{\perp \bi_{1}}\}$. 

For an integer $k \in \Z^{+}$ we use $[k]$ to denote the set $\{1, \ldots, k\}$. 

\subsection{Lattice Problems}

In this paper we are concerned with the following approximation problems, which are parametrized by some $\gamma > 1$.

\begin{description}
\item[$\gamma$-$\svp$:] Given a lattice basis $\B$, find a non-zero vector $\vi \in \La(\B)$ such that $\ltn{\vi} \leq \gamma\lambda_1(\B)$.

%\item[$\gamma$-$\sivp$:] Given a lattice basis $\B$, find a maximal set of linearly independent vectors ($n$, if $n$ is the rank of the matrix $\B$) of length at most $\gamma\lambda_n(\B)$.

\item[$\gamma$-$\cvp$:] Given a lattice basis $\B$, and a vector $\ti \in \R^m$ find a vector $\vi \in \La(\B)$ such that $\ltn{\vi-\ti} \leq \gamma \di(\ti,\B)$.
\end{description}

\noindent
We also use the following promise problems, which are parameterized by some $\gamma > 0$. 

\begin{description}
\item[$\gamma$-$\bdd$:] Given a lattice basis $\B$, and a vector $\ti \in \R^m$ with the promise that $\di(\ti,\La(\B)) \leq \gamma\lambda_1(\B)$, find a vector $\vi \in \La(\B)$ such that $\ltn{\vi-\ti} = \di(\ti,\B)$.

\item[$\gamma$-$\usvp$:] Given a lattice basis $\B$ with the promise that $\lambda_2(\B) \geq \gamma\lambda_1(\B)$, find a non-zero vector $\vi \in \La(\B)$ such that $\ltn{\vi} = \lambda_1(\B)$
(this makes sense only for $\gamma \geq 1$).
\end{description}

We assume that we have given a $\gamma$-$\svp$ oracle, denoted by $\Or$. 
When given a set of linearly independent vectors
 $\B = \{\bi_1, \bi_2, \dots, \bi_n\} \in \Q^{m\times n}$, $\Or(\B)$
 returns an elementary vector $\vi \in \La(\B)$
 which satisfies $0 < \ltn{\vi} \leq \gamma \lambda_1(\La(\B))$
(if $\vi$ is not elementary then we can find out the multiple and recover 
the corresponding elementary vector).

\section{Some basic tools}

Given a basis $\B$ and an elementary vector $\vi \in \La(\B)$, we can in polynomial time find a new basis of $\La(\B)$ of the form $\{\vi,\bi_2^{'}, \dots, \bi_n^{'}\}$. To do this we use the following lemma from Micciancio \cite{Mic08} (page~7, Lemma~1), which we specialized somewhat for our needs.

\begin{lemma}\label{newbasis}
There is a polynomial time algorithm $\fb(\vi,\B)$, which, 
on input an elementary vector $\vi$ of $\La(\B)$ and a lattice basis 
$\B \in \Q^{m\times n}$ outputs~$\tilde{\B} = (\tilde{\bi}_2,\ldots,
\tilde{\bi}_{n})$ 
such that $\La(\vi,\tilde{\bi}_2, \dots, \tilde{\bi}_{n}) = \La(\B)$.
\end{lemma}

\begin{lemma}\label{plb}
Let $\La(\B)$ be a lattice and $\vi \in \La(\B)$ be a vector in the lattice. If $\La(\B_{\perp \vi})$ is the projected lattice of $\La(\B)$ perpendicular to $\vi$ then $\lambda_i(\B_{\perp \vi}) \leq \lambda_{i+1}(\B)$, $i\in[n-1]$.
\end{lemma}
\begin{proof}
Let $\vi_i$ be the vector of length $\lambda_i(\B)$ such that $\{\vi_1, \dots, \vi_n\}$ are linearly independent. A set of such vectors exists \cite{GM02}. If $(\vi_1)_{\perp \vi} = 0$ then $(\vi_{i>1})_{\perp \vi} \in \La(\B_{\perp \vi})$ and $0 < \ltn{(\vi_i)_{\perp \vi}} \leq \ltn{\vi_i}$, proving the lemma. If $(\vi_1)_{\perp \vi} \neq 0$ then $(\vi_1)_{\perp \vi} \in \La(\B_{\perp \vi})$ and $0 < \ltn{(\vi_1)_{\perp \vi}} \leq \ltn{\vi_1}$. We argue in a similar way with $(\vi_2)_{\perp \vi}$ to prove the lemma for $i > 1$.
\qed\end{proof}

We use the following reduction from due to Lyubashevsky and
Micciancio~\cite{LM09}.
\begin{theorem}\label{lm}
For any $\gamma \geq 1$, there is a polynomial time oracle
reduction from $\bdd_{\frac{1}{2\gamma}}$ to $\usvp_\gamma$.
\end{theorem}
For completeness, we sketch a proof of Theorem \ref{lm} in
Appendix~\ref{app:lyubashevskyMicciancio}.

\section{Reducing $\cvp$ to $\svp$}
We prove the following theorem:
\begin{theorem}\label{main}
Given a basis $\B \in \Q^{m\times n}$ and a vector $\ti \in \R^m$, the problem $\gamma^2\sqrt{n}$-$\cvp$ is Turing reducible to the problem $\gamma$-$\svp$ in time $\poly(n,\log \gamma,\max_i \log \ltn{\bi_i})$.
\end{theorem}

In this section we give the algorithm to prove our theorem, and show that once it terminates, it satisfies the requirements of the theorem.  We will show that the algorithm runs in polynomial time in the next section.

The reduction takes as an input a basis $\B \in \Q^{m \times n}$ and a vector $\ti \in \R^m$. 
Recall that the oracle $\Or$ takes as  input a basis over $\Q$ and outputs an elementary vector which is a $\gamma$-approximation to the shortest vector. 
The reduction is given in Algorithm~\ref{cvpsolver}.

\begin{algorithm}
\caption{$\cvp(\B,\ti)$\hfill 
\emph{(input: $\B \in \Q^{m\times n}$, $\ti \in \Q^m$)}}
\label{cvpsolver}
\begin{algorithmic}[1]
\IF{$n=1$} 
\STATE Let $\bi_1$ be the only column of $\B$.
\STATE {\bf return} $a\bi_1$ with $a\in \Z$ such that $\|a \bi_1 - \ti\|$ is minimal.
\ELSE 
\STATE $\zi_1 \gets \frac{1}{2\gamma}$-$\bdd(\B,\ti)$ \COMMENT{Solve this with calls to $\Or$ as in Theorem~\ref{lm}}
\STATE $\vi \gets \Or(\B)$% \COMMENT{$\frac{\ltn{{\bf \vi}}}{\gamma} \leq \lambda_1 \leq \ltn{{\bf \vi}}$}
\STATE $\{\bi_{2},\dots, \bi_{n}\} \gets \LLL(\fb(\vi,\B))$
\STATE $\forall i \in \{2,\ldots,n\}: 
(\bi'_i)_{\perp \vi} \gets \bi_i-{\bi_i}|_{\vi}$
\STATE $\B^{'}_{\perp \vi} \gets \{(\bi'_2)_{\perp \vi}, \ldots, (\bi'_n)_{\perp \vi}\}$
\STATE $\ti'_{\perp} \gets \ti - \ti|_{\vi}$
\STATE $\zi'_2 \gets \cvp(\B^{'}_{\perp \vi},\ti'_{\perp \vi})$
\STATE Find $(a_2,\dots, a_n) \in \Z^{n-1}$ such that $\zi'_2=\sum_{i=2}^{n}a_i({\bi'_i})_{\perp \vi}$ 
\STATE Find $a_1 \in \Z$ such that $\zi_2=a_1 \vi + \sum_{i=2}^{n} a_i \bi_i$ is closest to $\ti$
\STATE {\bf return} the element of $\{\zi_1, \zi_2\}$ which is closest to $\ti$.
\ENDIF
\end{algorithmic}
\end{algorithm}
In line~6, we can simulate an oracle for $\frac{1}{2\gamma}$-\bdd{} due to 
Theorem~\ref{lm}, given $\Or$.
In line~7 we run the LLL algorithm on the basis returned by $\fb$; this 
is an easy way to ensure that the representation of the basis 
does not grow too large 
(cf.~the proof of Lemma~\ref{lem:polybits2}).
The optimization problem in line~13 is of course easy to
solve: for example, we can find $a_1' \in \mathbb{R}$ which minimizes 
the expression and then round $a_1'$ to the nearest integer.
\begin{theorem}\label{algoworks}
The approximate $\cvp$-solver {\tt (Algorithm \ref{cvpsolver})} outputs a vector $\zi \in \La(\B)$ such that $\ltn{\zi-\ti} \leq \gamma^2\sqrt{n}\di(\ti,\B)$.
\end{theorem}
\begin{proof}
We prove the theorem by induction on $n$. For the base case (i.e., $n=1$) we find the closest vector to $\ti$ in a single vector basis. This can be done exactly by finding the correct multiple of the only basis vector that is closest to $\ti$.

When $n>1$, we see that each run of the algorithm finds two candidates
$\zi_1$ and $\zi_2$.  
We show that the shorter of the two is an approximation to
the closest vector to $\ti$ in $\La(\B)$ for which
\begin{equation}\label{gcv}
\ltn{\zi-\ti} \leq \sqrt{n}\gamma^2\di(\ti,\B)
\end{equation}

We divide the proof in two cases, depending on whether $\di(\ti,\B) <
\frac{\lambda_1(\B)}{2\gamma}$.  
It is sufficient to show that in each case one of $\zi_1$ or $\zi_2$
satisfies Equation (\ref{gcv}).
\begin{enumerate}[1.]
\item If $\di(\ti,\B) < \frac{\lambda_1(\B)}{2\gamma}$, 
the promise of $\frac{1}{2\gamma}$-$\bdd$ is satisfied.
Thus, $\zi_1$ satisfies $\|\zi_1 - \ti\| \leq
\di(\ti,\B)$.
\item If $\di(\ti,\B) \geq \frac{\lambda_1(\B)}{2\gamma}$ we proceed
as in Kannan's proof to show that $\zi_2$ satisfies Equation (\ref{gcv}). 

By the induction hypothesis, $\zi_2'$ satisfies
\begin{align*}
\ltsqn{\zi_2'-\ti'_{\perp \vi}} \leq (n-1)\gamma^4\di^2(\ti'_{\perp \vi},\B'_{\perp \vi})
\end{align*}

At this point, note first that $\ti = \ti'_{\perp \vi} + \phi \vi$
for some $\phi\in \mathbb{R}$.
Since also $\sum_{i=2}^n a_i \bi_i = \zi_2' + \eta\vi$ for some
$\eta \in \mathbb{R}$, we can write
\begin{align*}
\| \zi_2 - \ti\|^2 &=
\ltsqn{(a_1 \vi + \zi_2' + \eta\vi) - (\ti'_{\perp \vi} + \phi\vi)}\\
&=
\|(a_1+ \eta - \phi)\vi\|^2 + \ltsqn{\zi'_2-\ti_{\perp \vi}}
\end{align*}
Since $a_1$ is chosen such that this expression is minimal we have
$|a_1 + \eta - \phi| \leq \frac12$, and so
\begin{align*}
\ltsqn{\zi_2-\ti} &\leq  
\ltsqn{\zi'_2-\ti_{\perp \vi}}+
\frac{\ltsqn{\vi}}{4}
\leq \ltsqn{\zi_2'-\ti_{\perp \vi}} + \frac{\gamma^2\lambda_1^2(\B)}{4}
\\&
\leq (n-1)\gamma^4\di^2(\ti_{\perp \vi},\La(\B_{\perp \vi})) + \frac{\gamma^2 4\gamma^2\di^2(\ti,\B)}{4}
\\&
\leq \gamma^4n\di^2(\ti,\B)\;.
\end{align*}
The second last inequality follows from $\lambda_1^2(\B) \leq 4\gamma^2\di^2(\ti,\B)$, which holds in this second case.  To see the last inequality, note that $\La(\B_{\perp \vi})$ is a projection of $\La(\B)$ and $\ti_{\perp \vi}$ is a projection of $\ti$ in the direction orthogonal to $\vi$, and a projection cannot increase the length of a vector. 
\end{enumerate}
Thus, in both cases one of $\zi_1$ and $\zi_2$ satisfies the requirements,
and so we get the result.
\qed\end{proof}

\section{Analysis of runtime}\label{polyp}
In this section, we show that Algorithm \ref{cvpsolver} runs in
polynomial time. 
Observe first that in each recursive call the number
of basis vector reduces by 1.
Since all steps are obviously polynomial, it is
enough to show that all the vectors generated
during the run of the algorithm can be represented in polynomially
many bits in the input size of the top level of the
algorithm. 
For this, we can assume that the original
basis vectors $\B=\{\bi_1, \dots, \bi_n\}$ are integer vectors. This can be achieved by multiplying them with the product of their denominators. This operation does not increase the bit representation by more than a factor of $\log (mn)$. Assuming that the basis vectors are over integers,
a lower bound on the input size can be given by $M=\max\{n,\log(\max_i
\ltn{b_i})\}$.

Given a basis $\B=\{\bi_1, \dots, \bi_n\}$, the Gram-Schmidt orthogonalization
of $\B$ is $\{\tilde{\bi}_1, \dots, \tilde{\bi}_n\}$, 
where $\tilde{\bi}_i=\bi_i-\sum_{j=1}^{i-1}\bi_i|_{\tilde{\bi}_j}$. 
We need the following Lemma from~\cite{GM02}. 

\begin{lemma}\label{lem:basisSize}{\em \cite{GM02}}
Let $\B = \{\bi_1,\ldots,\bi_n\}$ be $n$ linearly independent vectors.
Define the vectors $\tilde{\bi}_i = \bi_{i} - \sum_{j=1}^{i-1} \bi_i|_{\tilde{\bi_j}}$.
Then, the representation of any vector $\tilde{\bi}_i$ 
as a vector of quotients of natural numbers 
takes at most $\poly(M)$ bits for $M = \max\{n,\log(\max_{i}\|\bi_i\|)\}$.
\end{lemma}

\begin{lemma}\label{lem:vAsProjections}
Let $\vi_i$, $i \in [n]$, be the vector $\vi$ generated in the $i$th
level of the recursion in line~6 of Algorithm~\ref{cvpsolver}.

There is a basis $\xl_1,\ldots,\xl_n$ of $\B$
such that the vectors $\vi_i$ are given by the Gram-Schmidt orthogonalization
of $\xl_1,\ldots,\xl_{n}$. Furthermore, $\xl_1, \dots, \xl_n$ as well as $\vi_1, \dots, \vi_n$ 
are polynomially representable in $M$.
\end{lemma}
\begin{proof}
We first find lattice vectors $\xl_1, \dots, \xl_n \in \La(\B)$ which satisfy
\begin{align*}
\xl_i &= \vi_i + \sum_{j=1}^{i-1}\delta_j\vi_j
\end{align*}
for some $\delta_j \in [-\frac12, \frac{1}{2}]$, and then
 show that these vectors satisfy the claim of the lemma.

To see that such vectors exist, let $\B_j$ be the
 basis in the $j$th level of the recursion 
of Algorithm~\ref{cvpsolver}.  Then, we note that given a vector in $\La(\B_j)$ one can find a lattice vector in $\La(\B_{j-1})$ at distance at most $\frac{\ltn{\vi_{j-1}}}{2}$ in the direction of $\vi_{j-1}$ or $-\vi_{j-1}$.  We let $\xl_i$ be the vector obtained by doing such a lifting step repeatedly until we 
have a lattice vector in $\La(\B)$.

 The vectors $\vi_1, \dots, \vi_n$ are exactly the Gram-Schmidt orthogonalization of $\xl_1, \dots, \xl_n$,  because
\begin{align*}
\vi_i &= \xl_i - \xl_i|_{\vi_1} - \xl_i|_{\vi_2} - \dots - \xl_i|_{\vi_{i-1}}\;,
\end{align*}
and so the vectors $\xl_i$ must also form a basis of $\La(\B)$.

Also, we have for all $i \in [n]$:
\begin{align*}
\ltsqn{\xl_i} &\leq \ltsqn{\vi_i} + \frac{\ltsqn{\vi_{i-1}}}{4} + \dots + \frac{\ltsqn{\vi_1}}{4} \\
&\leq \sum_{j=1}^{i} \ltsqn{\vi_j}\\
&\leq n\gamma^2\lambda_n^2(\B)\tag{From Lemma \ref{plb}}
\end{align*}
As $\xl_1, \dots, \xl_n$ are vectors in the integer lattice $\B$; $\xl_1, \dots, \xl_n$ are polynomially representable in $M$ (and $\log \gamma$, but we can assume $\gamma < 2^{n}$). Coupled with Lemma~\ref{lem:basisSize} this completes the proof.
\qed
\end{proof}

\begin{lemma}\label{lem:polybits2}
All vectors which are generated in a run of Algorithm~\ref{cvpsolver}
have a representation of size $\poly(M)$ for $M=\max\{n,\log(\max_i
\ltn{b_i})\}$, in case the entries are represented as quotients of natural
numbers.
\end{lemma}
\begin{proof}
 The vectors $\vi_i$ which are generated in line~6 at different levels of 
recursion also have
representation of size $\poly(M)$ by Lemma~\ref{lem:basisSize}.
The basis $\B_i$ is $\LLL$ reduced and hence it is representable
in number of bits which is a fixed polynomial in the shortest
vector \cite{LLL82} and hence also $\vi_i$. 

The remaining vectors are produced by oracles which run in 
polynomial time or are small linear combinations of other vectors.
\qed
\end{proof}
We now give a proof of Theorem \ref{main}.
\begin{proof}(Theorem \ref{main}) Given $\B \in \Q^{m \times n}$ and $\ti \in \R^m$ we run Algorithm \ref{cvpsolver}. From Lemma \ref{algoworks}, the algorithm returns a vector $\zi$ which is a $\gamma^2\sqrt{n}$-approximation to the closest vector. Also, from Lemma~\ref{lem:polybits2}, all vectors in the algorithm have polynomial size representation, and so the algorithm runs in time ${\poly}(\log\gamma,M)$.
\qed\end{proof}

\section{Acknowledgements}

We thank Divesh Aggarwal and Robin K\"{unzler} for numerous helpful 
discussions, and Daniele Micciancio for useful comments on an earlier
version of this paper.
We want to thank the anonymous referees for helpful comments.  In particular,
we thank the referee who pointed out the relevance of Theorem~\ref{lm} to
our work and helped us simplify the proof to its current form.

\appendix
\section{Solving $\bdd$ using a $\usvp$-oracle}
\label{app:lyubashevskyMicciancio}
In this appendix we sketch the reduction from $\bdd_{1/2\gamma}$ to
$\usvp_{\gamma}$ from \cite{LM09} for completeness.
We will assume that $\di(\ti,\La(\B))$ is known -- it is shown in \cite{LM09}
how to avoid this assumption.
\begin{proof}(Theorem \ref{lm})
Let $(\B,\ti)$ be an instance of $\bdd_{\frac{1}{2\gamma}}$ and let $\alpha=\di(\ti,\La(\B)) \leq \frac{\lambda_1(\B)}{2\gamma}$. For simplicity we assume that we know $\alpha$ (see \cite{LM09} for bypassing this). Our goal is to find a vector $\ti^{\dagger} \in \La(\B)$ such that $\di(\ti^{\dagger},\ti) = \alpha$.
We define the new basis
\begin{equation}\label{deftildeb}
\tilde{\B} = \left(\begin{array}{cc}{\bf \B} & {\bf \ti}\\ {\bf 0} & \alpha\end{array}\right)\;.
\end{equation}
We will show that in $\tilde{\B}$ the vector $\vi := \big[\begin{array}{c}\ti^\dagger-\ti\\-\alpha\end{array}\big]$ is a  $\gamma$-unique shortest vector. 
It is clear that we can recover $\ti^{\dagger}$, the solution to the $\bdd$ problem, when given~$\vi$. The length of $\vi$ is $\sqrt{2}\alpha$, and so it is enough to show that all other vectors in $\La(\tilde{\B})$, which are not a multiple of $\vi$ have length at least $\sqrt{2}\gamma\alpha$.  Let us (for the sake of contradiction) assume that there is a vector $\vi_2$ of length at most $\|\vi_2\|<\sqrt{2}\gamma\alpha$ which is not a multiple of the vector $\vi$ above. We can write $\vi_2$ as $\vi_2 = \big[\begin{array}{c}\ui-a\ti\\-a\alpha\end{array}\big]$, where $\ui \in \La(\B)$ and $a \in \Z$. Since $\vi_2$ is not a multiple of $\vi$, it must be that $\ui - a\ti^\dagger \in \La(\B)$ is a non-zero lattice vector.  Now, using the triangle inequality, we get
\begin{align*}
\ltn{\ui-a\ti^\dagger} &\leq \ltn{\ui-a\ti} + a\ltn{\ti-\ti^\dagger} \\
&= \sqrt{\|\vi_2\|^2 - a^2\alpha^2} + a \alpha\\
&< \sqrt{2\alpha^2\gamma^2-a^2\alpha^2}+ a\alpha\\
&\leq 2\alpha\gamma\tag{Maximized when $a = \gamma$} \leq \lambda_1(\B)\;,
\end{align*}
which is a contradiction.
\qed
\end{proof}
\end{document}